\documentclass[10pt]{article}
\usepackage{amsfonts}
\usepackage{amsmath}
\usepackage{latexsym}

\usepackage{amsthm}
\usepackage[bookmarks]{hyperref}
\usepackage{comment}
\usepackage{graphics}
\usepackage{pict2e}

\usepackage{graphicx} 
\usepackage{amsfonts}

\usepackage{enumitem}
\usepackage{tikz}
\usetikzlibrary{intersections,positioning,shapes,calc,decorations.pathmorphing,decorations.pathreplacing}

\usepackage{color}
\usepackage{soul}
\definecolor{lightblue}{rgb}{.60,.60,1}

\def\nottoobig#1{{\hbox{$\left#1\vcenter to1.111\ht\strutbox{}\right.\n@space$}}}

\newtheorem{fact}{Fact}

\newtheorem{theorem}{Theorem}[section]

\newtheorem{lemma}[theorem]{Lemma}

\newtheorem{definition}[theorem]{Definition}

\usepackage{epsfig}
\usepackage{amsmath}

\newcommand{\leqp}{\leq^{+}}
\newcommand{\geqp}{\geq^{+}}
\newcommand{\eqp}{=^{+}}

\newcommand{\na}{n_A}
\newcommand{\nb}{n_B}
\newcommand{\nc}{n_C}
\newcommand{\pa}{p_A}
\newcommand{\pb}{p_B}
\newcommand{\pc}{p_C}
\newcommand{\xa}{x_A}
\newcommand{\xb}{x_B}
\newcommand{\xc}{x_C}
\newcommand{\ea}{E_A}
\newcommand{\eb}{E_B}
\newcommand{\ec}{E_C}
\newcommand{\ga}{G_A}
\newcommand{\gb}{G_B}
\newcommand{\gc}{G_C}
\newcommand{\zo}{\{0,1\}}
\newcommand{\mapping}{\rightarrow}
\newcommand{\ie}{$\mbox{i.e.}$}
\newcommand{\poly}{\rm{poly}}
\newcommand{\prob}{\rm{Prob}}

\makeatletter
\def\@listI{\leftmargin\leftmargini \parsep 4.5pt plus 1pt minus 1pt\topsep6pt plus 2pt minus 2pt \itemsep  2pt plus 2pt minus 1pt}

\let\@listi\@listI
\@listi
\makeatother

\author{ {Marius Zimand\/}
\thanks{  Department of Computer and Information Sciences, Towson University,
Baltimore, MD. http://triton.towson.edu/\~{ }mzimand}}

\title{Distributed compression through the lens of algorithmic information theory: a primer}

\begin{document}

\maketitle
\begin{abstract}
Distributed compression is the task of compressing correlated data by several parties, each one possessing  one piece of data and acting separately.  The classical Slepian-Wolf 
theorem~\cite{sle-wol:j:distribcompression} shows that if data is generated by independent draws from a joint distribution, that is by a memoryless stochastic process, then distributed compression can achieve the same compression rates as centralized compression when the parties act together. Recently, the author~\cite{zim:c:kolmslepianwolf} has obtained an analogue version of the Slepian-Wolf theorem in the framework of Algorithmic Information Theory (also known as Kolmogorov complexity).   The advantage over the classical theorem, is that the AIT version works for individual strings, without any assumption regarding the generative process. The only requirement is that the parties know the complexity profile of the input strings, which is a simple quantitative measure of the data correlation. The goal of this paper is to present in an accessible form that omits some technical details the main ideas from the reference~\cite{zim:c:kolmslepianwolf}.
\end{abstract}

\section{On busy friends wishing to share points}
\label{s:friends}
Zack has three good friends, Alice, Bob, and Charles, who share with him every piece of information they have. One day, Alice, Bob, and Charles, \emph{separately},  observe three \emph{collinear} points $A$, respectively $B$ and $C$, in the $2$-dimensional affine space over the field with $2^n$ elements. Thus, each one of Alice, Bob, and Charles possesses $2n$ bits of information, giving the two coordinates of their respective points. Due to the geometric relation, collectively, they have $5n$ bits of information, because given two points the third one can be described with just one coordinate. They want to email the points to Zack, without wasting bandwidth, that is by sending  approximately $5n$ bits, where ``approximately" means that they can afford an overhead of $O(\log n)$ bits.  Clearly, if they collaborate, they can send exactly $5n$ bits. The problem is that they have busy schedules, and cannot find  a good time to get together and thus they have to compress their points in isolation.  
How many bits do they need to send to Zack?  Let us first note some necessary requirements for the compression lengths.  Let $\na$ be the number of bits to which Alice compresses her point $A$, and let $\nb$ and $\nc$ have the analogous meaning for Bob and Charles.  It is necessary that 
\[\na + \nb + \nc \geq 5n, \]
because Zack needs to acquire $5n$ bits. It is also necessary that 
\[
\na + \nb \geq 3n, \na+\nc \geq 3n, \nb+\nc \geq 3n,
\]
 because if Zack gets somehow one of the three points, he still needs $3n$ bits of informations from the other two points. And it is also necessary that
 \[
 \na \geq n, \nb \geq n, \nc \geq n,
 \]
 because if Zack gets somehow two of the three points, he still needs $n$ bits of informations from the remaining point.

We will see that any numbers $\na, \nb$ and $\nc$ satisfying the above necessary conditions, are also sufficient up to a small logarithmic overhead, in the sense that there are probabilistic compression algorithms such that if $\na, \nb$ and $\nc$ satisfy these conditions, then Alice  can compress point $A$ to a binary string $\pa$ of length $\na + O(\log n)$, Bob can compress point $B$ to a binary string $\pb$ of length $\nb + O(\log n)$, Charles  can compress point $C$ to a binary string $\pc$ of length $\nc + O(\log n)$,  and Zack can with high probability reconstruct the three points from $\pa, \pb$ and $\pc$. Moreover, the compression does not use the geometric relation between the points, but only the correlation of information in the points, as expressed in the very flexible framework of algorithmic information theory.

\section{Algorithmic information theory}
\label{s:ait}
Algorithmic Information Theory (AIT), initiated independently by Solomonoff~\cite{sol:j:inductive}, Kolmogorov~\cite{kol:j:kolmcomplexity}, and Chaitin~\cite{cha:j:length-of-programs}, is a counterpart to the Information Theory (IT), initiated by Shannon. In IT the central object is a random variable $X$ whose realizations are strings over an alphabet $\Sigma$. The Shannon entropy of $X$ is defined by
\[
H(X) = \sum_{x \in \Sigma} P(X=x) ~(1/\log  P(X=x)).
\]
The entropy $H(X)$  is viewed as the amount of information in $X$, because each string $x$ can be described with $\lceil 1/\log P(X=x) \rceil$ bits (using the Shannon code), and therefore $H(X)$ is the expected number of bits needed to describe the outcome of the random process modeled by $X$.

AIT dispenses with the stochastic generative model, and defines the complexity of an individual string $x$ as the length of its shortest description. For example, the string $$x_1 = 00000000 00000000 00000000 00000000$$   has low complexity because it can be succinctly described as ``$2^5$ zeros." The string $$x_2 = 10110000 01010111 01010100 11011100$$ is a $32$-bit string obtained using random atmospheric noise (according to random.org), and has high complexity because it does not have a short description.

Formally, given a Turing machine $M$, a string $p$ is said to be a \emph{program} (or  a \emph{description}) of a string $x$, if $M$ on input $p$ prints $x$. We denote the length of a binary string $x$ by $|x|$. The \emph{Kolmogorov complexity} of $x$ relative to the Turing machine $M$ is
\[
C_M(x) = \min \{|p| \mid \mbox{ $p$ is a program for $x$ relative to $M$}\}.
\]
If $U$ is universal Turing machine, then for every other Turing machine $M$ there exists a string $m$ such that $U(m, p) = M(p)$ for all $p$, and therefore for every string $x$,
\[
C_U(x) \leq C_M(x) + |m|.
\]
Thus, if we ignore the additive constant $|m|$, the Kolmogorov complexity of $x$ relative to $U$ is minimal. We fix a universal Turing machine $U$, drop the subscript $U$ in $C_U(\cdot)$, and denote the complexity of $x$ by $C(x)$. We list below a few basic facts about Kolmogorov complexity:
\begin{enumerate}
\item For every string $x$, $C(x) \leq |x| + O(1)$, because  a string $x$ is trivially described by itself. (Formally,  there is a Turing machine $M$ that, for every $x$,  on input $x$ prints $x$.)
\item Similarly to the complexity of $x$, we define the complexity of $x$ conditioned by $y$ as
$C(x \mid y) = \min\{|p| \mid \mbox{ $U$ on input $p$ and $y$ prints $x$}\}.$
\item Using some standard pairing function $\langle \cdot, \cdot \rangle$ that maps pair of strings into single strings, we define $C(x,y)$ the complexity of a pair of strings  (and then we can extend to tuples with larger arity) by $C(x,y) = C(<x,y>)$.
\item We use the convenient shorthand notation $ a \leqp b$ to mean that $a \leq b + O(\log(a+b))$, where the constant hidden in the $O(\cdot)$ notation only depends on the universal machine $U$.  Similarly $a \geqp b$ means $a \geq b - O(\log(a+b))$, and $a \eqp b$ means ($a \leqp b$ and $a \geqp b$).
\item The chain rule in information theory states that $H(X,Y) = H(X) + H(Y \mid X)$. A similar rule holds true in algorithmic information theory: for all $x$ and $y$,
$C(x,y) \eqp C(x) + C(y \mid x)$.
\end{enumerate}
\section{Distributed compression, more formally}

We present the problem confronting Alice, Bob, Charles (the senders) and Zack (the receiver) in an abstract and formal setting. 
We assume that each one of Alice, Bob, and Charles has $n$ bits of information, which, in concrete terms, means that Alice has an $n$-bit binary string $\xa$,  Bob has an $n$-bit binary string $\xb$,  and Charles has an $n$-bit binary string $\xc$.  We also assume that the $3$-tuplet $(\xa, \xb, \xc)$ belongs to a set $S \subseteq \zo^n \times \zo^n \times \zo^n$, which defines the way in which the information is correlated (for example, $S$ may be the set of all three collinear points) and that all parties (\ie, Alice, Bob, Charles, and Zack) know $S$. Alice is using an encoding fumction $E_A :\zo^n \mapping \zo^{\na}$,  Bob is using an encoding fumction $E_B :\zo^n \mapping \zo^{\nb}$,  Charles is using an encoding fumction $E_C :\zo^n \mapping \zo^{\nc}$,  and Zack is using a decoding fumction $D :\zo^{\na} \times \zo^{\nb} \times \zo^{\nc} \mapping \zo^{n}$.  Ideally, the requirement is that for all $(\xa, \xb, \xc)$ in $S$, $D(E_A(\xa), E_B(\xb), E_C(\xc)) = (\xa, \xb, \xc)$.  However, since typically the encoding functions are probabilistic, we allow the above equality to fail with probability bounded by some small $\epsilon$, where the probability is over the random bits used by the encoding functions. Also, sometimes, we will be content if the encoding/decoding procedures work, not for all, but only for ``most" $3$-tuples in $S$ (\ie, with probability close to $1$, under a given probability distribution on $S$).

Our focus in this paper is to  present distributed compression in the framework of Algorithmic Information Theory, but let us present first  the point of view of Information Theory,  where the problem  has been studied early on. The celebrated classical theorem of Slepian and Wolf~\cite{sle-wol:j:distribcompression} characterizes the possible compression rates $\na, \nb$ and $\nc$ for the case of \emph{memoryless} sources. The memoryless assumption means that $(\xa, \xb, \xc)$ are realizations of random variables $(X_A, X_B, X_C)$, which consist of $n$ independent  copies of a random variable that has a joint distribution $P(b_1, b_2, b_3)$ on triples of bits. In other words, the generative model for $(\xa, \xb, \xc)$ is a stochastic process that consists  of $n$ independent draws from the joint distribution, such that Alice observes $\xa$, the sequence of first components in the $n$ draws, Bob observes $\xb$, the second components, and Charles observes  $\xc$, the third components.  A stochastic process of this type is called 3-DMS (Discrete Memoryless Source). By Shannon's Source Coding Theorem, if  $n'$ is a number that is at least $H(X_A, X_B, X_C)$  and if Alice, Bob, and Charles put their data together, then, for every $\epsilon > 0$, there exists an encoding/decoding pair $E$ and $D$, where $E$  compresses $3n$-bit strings to $(n' + \epsilon n)$-bit strings and $D(E(X_A, X_B,X_C)) = (X_A, X_B, X_C)$ with probability $1-\epsilon$, provided $n$ is large enough.  The second part of Shannon's Source Coding Theorem shows that this is essentially optimal because if the data is compressed to length smaller than $H(X_A,X_B,X_C) -\epsilon n$ (for constant $\epsilon$), then the probability of correct decoding  goes to $0$. The Slepian-Wolf Theorem shows that such a compression can also be done if Alice, Bob, and Charles compress separately. Actually, it describes precisely the possible compression lengths. Note that if the three senders  compress separately to lengths $\na, \nb$ and $\nc$ as indicated above, then it is essentially necessary that $\na + \nb  + \nc \geq H(X_A, X_B, X_C) - \epsilon n$, $\na + \nb \geq  H(X_A, X_B \mid X_C) - \epsilon n$ (because even if Zack has $X_C$, he still needs to receive a number of bits equal to the amount of entropy in $X_A$ and $X_B$ conditioned  by $X_C$), $\na  \geq H(X_A \mid X_B, X_C) - \epsilon n$ (similarly,  even if Zack has $X_B$ and $X_C$, he still needs to receive a number of bits equal to the amount of entropy in $X_A$ conditioned by $X_B$ and $X_C$), and there are the obvious other necessary conditions obtained by permuting $A, B$ and $C$. The Slepian-Wolf Theorem shows that for 3-DMS these necessary conditions are, essentially, also sufficient, in the sense that the slight change of $-\epsilon n$ into $+\epsilon  n$ allows encoding/decoding procedures. Thus, in the next theorem we suppose that $\na + \nb  + \nc \geq H(X_A, X_B, X_C) + \epsilon n$, and similarly for the other relations.
\begin{theorem}[Slepian-Wolf Theorem~\cite{sle-wol:j:distribcompression}]
Let $(X_A, X_B, X_C)$ be a 3-DMS, let $\epsilon > 0$,  and let $\na, \nb, \nc$ satisfy the above  conditions (with $+\epsilon n$ instead of $-\epsilon n$). Then there exist encoding functions $E_A : \zo^n \mapping \zo^{\na}, E_B: \zo^n \mapping \zo^{\nb}, E_C : \zo^n \mapping \zo^{\nc}$ and a decoding function $D: \zo^{\na} \times \zo^{\nb} \times \zo^{\nc} \mapping \zo^n$ such that $Prob [D(E(X_A, X_B, X_C)) = (X_A,X_B,X_C)] \geq 1- O(\epsilon)$, provided $n$ is large enough.
\end{theorem}

There is nothing special about three senders, and indeed the Slepian-Wolf theorem holds for any number $\ell$ of senders, where $\ell$ is a constant, and for sources which are $\ell$-DMS over any alphabet $\Sigma$. This means that the senders compress $(x_1, \ldots, x_\ell)$, which is realization of random variables $(X_1, \ldots, X_\ell)$, obtained  from $n$ independent draws from a joint distribution $p(a_1, \ldots, a_\ell)$, with each $a_i$ ranging over the alphabet $\Sigma$. The $i$-th sender observes the realization $x_i$  of $X_i$, and uses an encoding function $E_i : \Sigma^n \mapping \Sigma^{n_i}$.  Suppose that the compression lengths $n_i$, $i=1, \ldots, \ell$, satisfy $\sum_{i \in V} n_i \geq H(X_V \mid X_{\overline{V}}) + \epsilon n$, for every subset $V \subseteq \{1, \ldots, \ell\}$ (where if $V= \{i_1, \ldots, i_t\}$, $X_V$ denotes the tuple $(X_{i_1}, \ldots, X_{i_t})$, and $\overline{V}$ denotes $\{1, \ldots, \ell\} - V$).  Then the Slepian-Wolf theorem for $\ell$-DMS states  that there are $E_1, \ldots, E_\ell$ of the above type, and $D: \Sigma^{n_1} \times \ldots \times \Sigma^{n_\ell} \mapping \Sigma^n$ such that
$D(E_1(X_1), \ldots, E_\ell(X_\ell)) = (X_1, \ldots, X_\ell)$ with probability $1-\epsilon$.

As pointed out above, the Slepian-Wolf theorem shows the surprising and remarkable fact that, for memoryless sources, distributed compression can be done at an optimality level that is on a par with centralized compression. On the weak side, the memoryless property means that there is a lot of independence in the generative process: the realization at time $i$ is independent of the realization at time $i-1$. Intuitively, independence helps distributed compression. For example, in the limit case in which the senders observe realizations of fully independent random variables, then, clearly, it makes no difference whether compression is distributed or centralized. The Slepian-Wolf theorem has been extended to sources that are stationary and ergodic~\cite{cov:j:slepwolfergodic}, but these sources are still quite simple, and intuitively realizations which are temporally sufficiently apart are close to being independent.

One may be inclined to believe that the optimal compression rates of distributed compression in the theorem are caused by the independence properties of the sources. However, this is not so, and  we shall see that in fact the Slepian-Wolf phenomenon does not require any type of independence.  Even more, it works without any generative model. For that we need to work in the framework of Kolmogorov complexity (AIT). 

Let us recall the example from Section~\ref{s:friends}: Alice, Bob, and Charles observe separately, respectively, the collinear points $A, B, C$. Even without assuming any generative process for the three points, we can still express their correlation using Kolmogorov complexity. More precisely, their correlation is described by the Kolmogorov complexity profile, which consists of $7$ numbers, giving the complexities of all  non-empty subsets of $\{A, B, C \}$: $(C(A), C(B), C(C), C(A,B), C(A,C), C(B,C), C(A,B,C))$.

Let us consider the general case, in which the three senders have, respectively, $n$-bit strings $\xa, \xb,\xc$ having a given complexity profile $(C(x_V) \mid V\subseteq \{\xa,\xb,\xc\}$ , $V \not= \emptyset)$ (where $x_V$ is the notation convention that we used for the $\ell$-senders case of the Slepian-Wolf theorem).  What are the possible compression lengths, so that Zack can decompress and obtain $(\xa, \xb, \xc)$ with probability $(1-\epsilon)$?

To answer this question, for simplicity, let us consider the case of a single sender, Alice. She  wants to use a probabilistic encoding function $E$ such that there exists a decoding function $D$ with the property that for all $n$, and for all $n$-bit strings $x$, $D(E(x)) = x$, with probability $1-\epsilon$. A lower bound on the length $|E(x)|$ is given in the following lemma.
\begin{lemma}
Let $E$ be a probabilistic encoding function, and $D$ be a decoding function such that for all  strings $x$, $D(E(x)) = x$, with probability $(1-\epsilon)$. Then for every $k$, there is a string $x$ with $C(x ) \leq k$, such that $|E(x)| \geq k + \log(1-\epsilon) - O(1)$.
\end{lemma}
\begin{proof}
Fix $k$ and let $S= \{x  \mid C(x ) \leq k\}$. It  can be shown that for some constant $c$, $|S| \geq 2^{k-c}$, where $|S|$ is the size of $S$ (the idea is that the first string which is not in $S$ can be described with $\log |S| + O(1)$ bits).  Since for every $x \in S$, $D(E(x, \rho))= x$ with probability $1-\epsilon$ over the randomness $\rho$, there is some fixed randomness $\rho$ such that $D(E(x,\rho))=x$, for a fraction of $1-\epsilon$ of the $x$'s in $S$. Let $S' \subseteq S$ be the set of such strings $x$.  Thus, $|S'| \geq (1-\epsilon)|S| \geq (1-\epsilon)2^{k-c}$ and the function $E(\cdot,\rho)$ is one-to-one on $S'$ (otherwise decoding would not be possible). Therefore the function $E(\cdot, \rho)$ cannot map all $S'$ into strings of length $k + \log(1-\epsilon) - (c+1)$.
\end{proof}
In short, if for every $x$, $D(E(x))= x$ with probability $1-\epsilon$, then for infinitely many $x$ it must be the case that $|E(x)| \geq C(x) + \log (1-\epsilon) - O(1)$. In other words, if we ignore the small terms, it is not possible to compress to length less than $C(x)$.

In the same way, similar lower bounds can be established for the case of more senders. For example, let us consider three senders that use the probabilistic encoding functions $E_A, E_B$ and $E_C$:  if there is a decoding function $D$ such that for every $(\xa,\xb,\xc)$,
\[
D(E_A(\xa), E_B(\xb), E_C(\xc))  = (\xa, \xb, \xc), \mbox{~~with  probability $1-\epsilon$},
\]
(where the probability is over the randomness used by the encoding procedures) then for infinitely many $(\xa,\xb,\xc)$
\[
\begin{array}{rl}
|E_A(\xa)| + |E_B(\xb)| +|E_C(\xc)| &\geq C(\xa,\xb,\xc) + \log(1-\epsilon) - O(1), \\
|E_A(\xa)| + |E_B(\xb)|  &\geq C(\xa,\xb\mid \xc) + \log(1-\epsilon) - O(1), \\
|E_A(\xa)| &\geq C(\xa \mid \xb,\xc) + \log(1-\epsilon) - O(1),\\
\end{array}
\]
and similar relations hold for any permutation of $A, B$ and $C$. As we did above, it is convenient to use the notation convention that if $V$ is a subset of $\{A,B,C\}$, we let $x_V$ denote the tuple of strings with indices in $V$ (for example, if $V=\{A,C\}$, then $x_V=(\xa,\xc)$). Then the above relations can be written concisely as
\[
\sum_{i \in V} |E_i(x_i)| \geq C(x_V \mid x_{\{A,B,C\}-V}) + \log(1-\epsilon) - O(1), \mbox{ for all $V \subseteq \{A,B,C\}$},
\]

The next theorem is the focal point of this paper. It shows that the above necessary conditions regarding the compression lengths are, essentially, also sufficient. 
\begin{theorem}[Kolmogorov complexity version of Slepian-Wolf coding~\cite{zim:c:kolmslepianwolf} ]
\label{t:kolmslepwolf}
There exist probabilistic algorithms $E_A, E_B, E_C$, a deterministic algorithm $D$,  and a function $\alpha(n) = O(\log  n)$ such  that 
for every $n$, for every tuple of integers $(\na,\nb,\nc)$, and for every tuple of $n$-bit strings $(\xa, \xb, \xc)$ if
\begin{equation}
\label{e:constraint}
\sum_{i \in V} n_i \geq C(x_V \mid x_{\{A,B,C\}-V}), \mbox{ for all $V \subseteq \{A,B,C\}$},
\end{equation}
then
\begin{itemize}
\item[(a)] $E_A$ on input $\xa$ and $\na$ outputs a string $\pa$ of length at most $\na + \alpha(n)$, $E_B$ on input $\xb$ and $\nb$ outputs a string $\pb$ of length at most $\nb+ \alpha(n)$,  $E_C$ on input $\xc$ and $\nc$ outputs a string $\pc$ of length at most $\nc + \alpha(n)$, 
\item[(b)]  $D$ on input $(\pa,\pb, \pc)$ outputs $(\xa, \xb, \xc)$, with probability $1-1/n$.

\end{itemize}
\end{theorem}

We present the proof of this theorem in the next section, but for now, we make several remarks:
\begin{itemize}
\item Compression procedures for \emph{individual} inputs (\ie, without using any knowledge regarding the generative process) have been previously designed using the celebrated Lempel-Ziv methods~\cite{lem-ziv:j:compress,ziv:j:compressindivid}. Such methods have been used for distributed compression as well~\cite{ziv:j:compresshelper,due-wol:j:slep-wolf-individ,kas:j:slep-wolf-individ}. For such procedures two kinds of optimality have been established, both valid for infinite sequences and thus having an asymptotic nature. First,  the procedures achieve a compression length that is asymptotically equal to the so-called finite-state complexity,  which is the minimum length that can be  achieved by finite-state encoding/decoding procedures.  Secondly, the compression rates are asymptotically optimal in case the infinite sequences are generated by sources that are stationary and ergodic~\cite{wyn-ziv:j:compress}. In contrast, the compression in Theorem~\ref{t:kolmslepwolf} applies to finite strings and achieves a compression length close to minimal description length. On the other hand, the Lempel-Ziv approach has lead to efficient compression algorithms that are used in practice.
 \item At the cost of increasing  the ``overhead" $\alpha(n)$ from $O(\log n)$ to $O(\log^3 n)$, we can obtain compression procedures $E_A, E_B$ and $E_C$ that run in polynomial time. On the other hand, the decompression procedure $D$ is slower than any computable function. This is unavoidable at this level of optimality (compression at close to minimum description length) because of the existence of \emph{deep strings}. (Informally, a string $x$ is deep  if it  has a description $p$ of small length but the universal machine takes a long time to produce $x$ from $p$.)
 \item The theorem is true for any number $\ell$ of senders, where $\ell$ is an arbitrary constant. We have singled out $\ell = 3$ because this case allows us to present the main ideas of the proof in a relatively simple form.
\item Romashchenko~\cite{rom:j:slepwolf} (building on an earlier result of Muchnik~\cite{muc:j:condcomp}) has obtained a Kolmogorov complexity version of Slepian-Wolf, in which the encoding and the decoding functions use $O(\log n)$ of extra information, called \emph{help bits}.    The above theorem eliminates the help bits, and is, therefore, fully effective.  The cost is that the encoding procedure  is probabilistic and thus  there is a small error probability. The proof of Theorem~\ref{t:kolmslepwolf} is inspired from Romashchenko's approach, but  the technical machinery is quite different.

 \item The classical Slepian-Wolf theorem. can be obtained from the Kolmogorov complexity version because if $X$ is memoryless, then with probability $1-\epsilon$, $H(X) - c_\epsilon \sqrt{n} \leq C(X) \leq H(X) + c_\epsilon \sqrt{n}$, where $c_\epsilon$ is a constant that only depends on $\epsilon$.

\end{itemize}

\section{Proof sketch of Theorem~\ref{t:kolmslepwolf}}
The central piece in the proof is a certain type of bipartite graph with a low congestion property.  We recall that in a bipartite graph, the nodes are partitioned in two sets, $L$ (the left nodes) and $R$ (the right  nodes), and all edges connect a left node to a right node. We allow multiple edges between two nodes. In the graphs that we use, all  left nodes have the same degree, called the left degree.
Specifically, we use bipartite graphs $G$ with $L=\zo^{n}$, $R=\zo^m$ and with left degree $D=2^d$. We label the edges outgoing from $x \in L$ with strings $y \in \zo^d$. We typically work with a family
of graphs indexed on $n$ and such a family of graphs is
\emph{computable} if there is an algorithm that on input $(x,y)$, where
$x \in L$ and $y \in \zo^d$, outputs the $y$-th neighbor of
$x$. Some of the graphs also depend on a rational $0 < \delta < 1$. A
constructible family of graphs is \emph{explicit} if the above algorithm
runs in time $\poly(n, 1/\delta)$.

We now introduce informally the notions of a \emph{rich owner} and of a \emph{graph with the rich owner property}. Let $B \subseteq L$. The $B$-degree of a right node is the number of its neighbors that are in $B$.  Roughly speaking a left node is a rich owner with respect to $B$,  if most of its right neighbors are ``well-behaved," in the sense that their $B$-degree is not much larger than  $|B| \cdot D /|R|$, the average right degree when the left side is restricted to $B$. One particularly interesting  case, which  is used many times in the proof,   is when most of the neighbors of a left $x$ have $B$-degree $1$, \ie, when $x$ ``owns" most of its right neighbbors. A  graph has the rich owner property if,  for all $B \subseteq L$, most of the left nodes in $B$ are rich owners with respect to $B$. In the formal definition below, we replace the average right degree with a value which may look arbitrary, but since in applications, this value is approximately equal to  the average right degree, the  above intuition should be helpful.

The precise definition of rich ownership  depends on two parameters $k$ and $\delta$. 

\begin{definition}  Let $G$ be a bipartite graph as above and let $B$ be a subset of $L$. We say that
$x \in B$ is a $(k,\delta)$-rich owner with respect to $B$ if the following holds:
\begin{itemize}
\item \emph{small regime case:}  If $|B| \leq 2^k$,  then at least $1-\delta$ fraction of $x$'s neighbors  have $B$-degree equal to $1$, that is they are not shared with any other nodes in $B$. We also say that $x \in B$  owns $y$ with respect to B  if $y$ is a neighbor of $x$ and the $B$-degree of $y$ is $1$.
\item  \emph{large regime case:}  If $|B| >  2^k$, then at least a $1-\delta$ fraction of $x$'s neighbors have $B$-degree at most $(2/\delta^2) |B| \cdot D /2^k$.
\end{itemize} 
If $x$
is not a $(k, \delta)$-rich owner with respect to $B$, then it is said to be a $(k,\delta)$-poor owner with respect to $B$.
\end{definition}
\begin{definition}
\label{d:ro}
A bipartite graph $G = (L=\zo^n , R=\zo^m, E \subseteq L \times R)$ has the $(k,\delta)$-rich owner property if 
for every set $B \subseteq L$  all nodes  in $B$, except at most $\delta |B|$ of them,  
 are $(k, \delta)$-rich owners with respect to $B$.
\end{definition}

The following theorem provides the type of graph that we use. 
\begin{theorem}
\label{t:richownergraph}
For every natural numbers $n$ and $k$ and for every rational number $\delta \in (0,1]$, there exists a computable   bipartite graph $G = (L,R, E \subseteq L \times R)$ that has the $(k,\delta)$-rich property with the following parameters: $L = \zo^n$, $R = \zo^{k+ \gamma(n/\delta)}$,  left degree $D = 2^{\gamma(n/\delta)}$, 
where $\gamma(n) = O(\log n)$.

There also exists an \emph{explicit} bipartite graph with the same parameters except that the overhead is $\gamma(n) = O(\log^3 n)$.

\if01
\begin{enumerate}
\item $L = \zo^n$,
\item  $R = \zo^{k+ \gamma(n/\delta)}$,
\item left degree $D = 2^{\gamma(n/\delta)}$, 
\end{enumerate}
where $\gamma(n) = O(\log^3(n/\delta))$.
\fi
\end{theorem}

The graphs in Theorem~\ref{t:richownergraph} are derived from randomness extractors. The computable graph is obtained with the probabilistic method, and we sketch the construction in 
Section~\ref{s:graph}. The explicit graph relies on the extractor from~\cite{rareva:c:extractor} and uses a combination of techniques from~\cite{raz-rei:c:extcon},~\cite{cap-rei-vad-wig:c:conductors}, and~\cite{bau-zim:c:linlist}. 
\begin{figure}[h]
\hspace{13ex}
\begin{tikzpicture}[shorten >=1pt,scale=0.45,auto,node distance=3cm, transform shape,
 NodeL/.style={draw, ellipse, minimum height=14cm, minimum width =48mm},
 NodeLL/.style={draw, ellipse, minimum height=6cm, minimum width =24mm},
NodeR/.style={draw, ellipse, fill=lightgray, minimum height=1mm, minimum width =1mm},
NodeRshared/.style={draw, ellipse, fill=white, minimum height=1mm, minimum width =1mm},
NodeRR/.style={draw, ellipse,minimum height=10cm, minimum width=24mm},
 myarrow/.style={->, shorten >=1pt, thick},
 tline/.style={-,shorten >=1pt, dashed, thick}
]
 
\node[NodeL] (leftset) at (-6,0){};
\node[anchor=south] at (leftset.90) {\LARGE  $L =\zo^n$};
\node[NodeLL] (leftset) at (-6,-1){};
\node[anchor=south] at (leftset.90) {\LARGE  $B$};
\node[] (leftnode) at (-6,-1){};
\node[anchor=south] at (leftnode.90) {\LARGE $x$};

\node[](mark) at (-6.0,1.0){};
\node[](mark2) at (-6.0,-3.0){};
\node[align=center, ellipse callout, draw, callout absolute pointer= (mark.west), above right=2.5cm of mark.north west, fill=black!20] {\LARGE rich owners};
\node[align=center, ellipse callout, draw, callout absolute pointer= (mark2.west), below right=2.5cm of mark2.north west, fill=black!20] {\LARGE poor  owners};

\begin{scope}[xshift=9.5cm]
  
\node[NodeRR](bigR) at (0,0){};
\node[anchor=south] at (bigR.90) {\LARGE  $R =\zo^{k+\gamma(n/\delta)}$};

 \end{scope}

\draw[](-7.15,-2.5)--(-4.75,-2.5);
\draw[dashed](-6.9,-2.5)--(-6.9,-3.2);
\draw[dashed](-6.6,-2.5)--(-6.6,-3.6);
\draw[dashed](-6.3,-2.5)--(-6.3,-3.9);
\draw[dashed](-6.0,-2.5)--(-6.0,-4.2);
\draw[dashed](-5.7,-2.5)--(-5.7,-3.9);
\draw[dashed](-5.4,-2.5)--(-5.4,-3.6);
\draw[dashed](-5.1,-2.5)--(-5.1,-3.2);


\node[] at (-6,-1) (x){};
\node[NodeR] at (9.5,1.3) (z1){};
\node[NodeR] at (9.5,0.8) (z2){};
\node[NodeR] at (9.5,0.3) (z3){};
\node[NodeRshared] at (9.5,-0.2) (z4){};
\node[] at (-6,-1.8) (dummy){};
\path[] (x) edge [bend left=10] (z1);
\path[] (x) edge [bend left=5] (z2);
\path[] (x) edge [bend left=0] (z3);
\path[] (x) edge [bend right= 5] (z4);
\path[tline] (dummy) edge [bend right = 15] (z4);
\if01
\draw[] (-6,-1) -- (9.5,1);
\draw [] (-6,-1) -- (9.5,0.5);
\draw [] (-6,-1) -- (9.5,0);
\draw [] (-6,-1) -- (9.5,-0.5);
\draw [] (-6,-1) -- (9.5,-1);
\draw[tline] (-6,-1.8) --(9.5,-1);
\fi
\draw[](-4,-0.3) edge [bend left=5] (-4,-1.25);

\node at (-2.0, 1)(deg) {\LARGE degree $D$};
\draw[->] (deg) edge (-4,-0.8);
\end{tikzpicture}
\caption{Graph with the $(k, \delta)$ rich owner property.   If $|B| \leq 2^k$ (small regime), a left node $x$ is a rich owner with respect to $B$ if it owns $(1-\delta)$ of its neighbors;  if $|B| > 2^k$ (large regime), if $(1-\delta)$ of its neighbors have $B$-degree close to the average right $B$-degree.  For every $B \subseteq L$, $(1-\delta)$ fraction of $B$ are rich owners.In the figure, the grey neighbors are owned by $x$, and the white neighbor is not owned.}
\label{f:grichowner}
\end{figure}
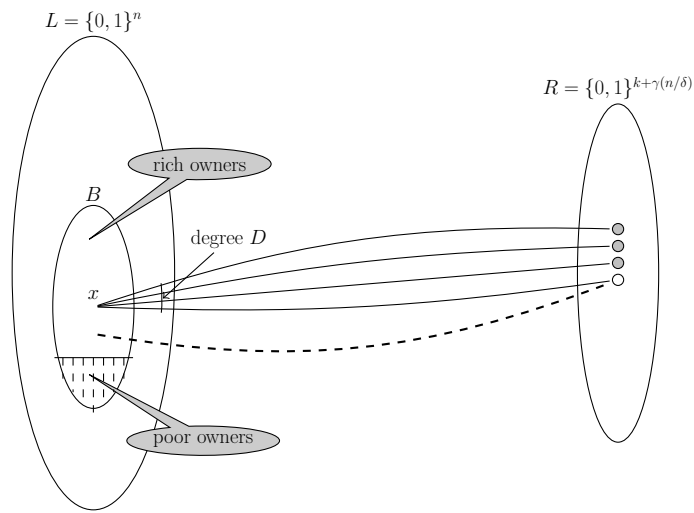

Let us proceed now to the proof sketch of Theorem~\ref{t:kolmslepwolf}. \emph{We warn the reader that for the sake of readability, we skip several technical elements. In particular, we ignore the loss of precision in $\eqp, \leqp, \geqp$, and we treat these relations as if they were $=, \leq, \geq$.}

 Recall that the input procedures $\ea, \eb$ and $\ec$ have as inputs, respectively, the pairs $(\xa,\na), (\xb,\nb), (\xc,\nc)$, where $\xa,\xb, \xc$ are $n$-bit strings, and $\na, \nb,\nc$ are natural numbers. The three encoding procedures use, respectively the graphs $\ga, \gb$ and $\gc$, which have, respectively, the $(\na+1, 1/n^2)$, $(\nb+1,1/n^2)$, $\nc+1,1/n^2)$ rich owner property. Viewing the strings $\xa,\xb, \xc$ as left nodes in the respective graphs, the encoding procedures pick $\pa,\pb, \pc$ as random neighbors of $\xa,\xb, \xc$ (see Figure~\ref{f:encode}).

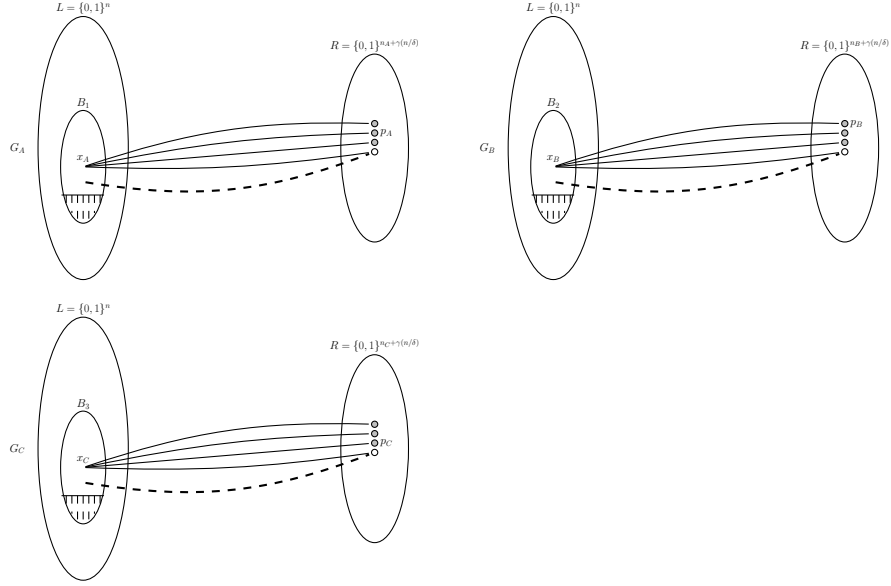
\begin{figure}[h]
\hspace{13ex}
\begin{tikzpicture}[shorten >=1pt,scale=0.25,auto,node distance=1cm, transform shape,
 NodeL/.style={draw, ellipse, minimum height=14cm, minimum width =48mm},
 NodeLL/.style={draw, ellipse, minimum height=6cm, minimum width =24mm},
NodeR/.style={draw, ellipse, fill=lightgray, minimum height=1mm, minimum width =1mm},
NodeRshared/.style={draw, ellipse, fill=white, minimum height=1mm, minimum width =1mm},
NodeRR/.style={draw, ellipse,minimum height=10cm, minimum width=36mm},
 myarrow/.style={->, shorten >=1pt, thick},
 tline/.style={-,shorten >=1pt, dashed, thick}
]
 
\node[NodeL] (leftset) at (-6,0){};
\node[anchor=south] at (leftset.90) {\LARGE  $L =\zo^n$};
\node[anchor=east] at (leftset.180) {\LARGE  {\bf $G_A$}\quad\quad};
\node[NodeLL] (leftset) at (-6,-1){};
\node[anchor=south] at (leftset.90) {\LARGE  $B_1$};
\node[] (leftnode) at (-6,-1){};
\node[anchor=south] at (leftnode.90) {\LARGE $\xa$};

\begin{scope}[xshift=9.5cm]
  
\node[NodeRR](bigR) at (0,0){};
\node[anchor=south] at (bigR.90) {\LARGE  $R =\zo^{\na+\gamma(n/\delta)}$};

 \end{scope}

\draw[](-7.15,-2.5)--(-4.75,-2.5);
\draw[dashed](-6.9,-2.5)--(-6.9,-3.2);
\draw[dashed](-6.6,-2.5)--(-6.6,-3.6);
\draw[dashed](-6.3,-2.5)--(-6.3,-3.9);
\draw[dashed](-6.0,-2.5)--(-6.0,-4.2);
\draw[dashed](-5.7,-2.5)--(-5.7,-3.9);
\draw[dashed](-5.4,-2.5)--(-5.4,-3.6);
\draw[dashed](-5.1,-2.5)--(-5.1,-3.2);


\node[] at (-6,-1) (x){};
\node[NodeR] at (9.5,1.3) (z1){};
\node[NodeR] at (9.5,0.8) (z2){};
\node[anchor=west] at (z2) {\LARGE \,\,$\pa$};
\node[NodeR] at (9.5,0.3) (z3){};
\node[NodeRshared] at (9.5,-0.2) (z4){};
\node[] at (-6,-1.8) (dummy){};
\path[] (x) edge [bend left=10] (z1);
\path[] (x) edge [bend left=5] (z2);
\path[] (x) edge [bend left=0] (z3);
\path[] (x) edge [bend right= 5] (z4);
\path[tline] (dummy) edge [bend right = 15] (z4);
\if01
\draw[] (-6,-1) -- (9.5,1);
\draw [] (-6,-1) -- (9.5,0.5);
\draw [] (-6,-1) -- (9.5,0);
\draw [] (-6,-1) -- (9.5,-0.5);
\draw [] (-6,-1) -- (9.5,-1);
\draw[tline] (-6,-1.8) --(9.5,-1);
\fi


 \begin{scope}[xshift=25cm]
\node[NodeL] (leftset) at (-6,0){};
\node[anchor=south] at (leftset.90) {\LARGE  $L =\zo^n$};
\node[anchor=east] at (leftset.180) {\LARGE  {\bf $G_B$}\quad\quad};
\node[NodeLL] (leftset) at (-6,-1){};
\node[anchor=south] at (leftset.90) {\LARGE  $B_2$};
\node[] (leftnode) at (-6,-1){};
\node[anchor=south] at (leftnode.90) {\LARGE $\xb$};

\begin{scope}[xshift=9.5cm]
  
\node[NodeRR](bigR) at (0,0){};
\node[anchor=south] at (bigR.90) {\LARGE  $R =\zo^{\nb+\gamma(n/\delta)}$};

 \end{scope}

\draw[](-7.15,-2.5)--(-4.75,-2.5);
\draw[dashed](-6.9,-2.5)--(-6.9,-3.2);
\draw[dashed](-6.6,-2.5)--(-6.6,-3.6);
\draw[dashed](-6.3,-2.5)--(-6.3,-3.9);
\draw[dashed](-6.0,-2.5)--(-6.0,-4.2);
\draw[dashed](-5.7,-2.5)--(-5.7,-3.9);
\draw[dashed](-5.4,-2.5)--(-5.4,-3.6);
\draw[dashed](-5.1,-2.5)--(-5.1,-3.2);


\node[] at (-6,-1) (x){};
\node[NodeR] at (9.5,1.3) (z1){};
\node[anchor=west] at (z1) {\LARGE \,\,$\pb$};
\node[NodeR] at (9.5,0.8) (z2){};
\node[NodeR] at (9.5,0.3) (z3){};
\node[NodeRshared] at (9.5,-0.2) (z4){};
\node[] at (-6,-1.8) (dummy){};
\path[] (x) edge [bend left=10] (z1);
\path[] (x) edge [bend left=5] (z2);
\path[] (x) edge [bend left=0] (z3);
\path[] (x) edge [bend right= 5] (z4);
\path[tline] (dummy) edge [bend right = 15] (z4);
\if01
\draw[] (-6,-1) -- (9.5,1);
\draw [] (-6,-1) -- (9.5,0.5);
\draw [] (-6,-1) -- (9.5,0);
\draw [] (-6,-1) -- (9.5,-0.5);
\draw [] (-6,-1) -- (9.5,-1);
\draw[tline] (-6,-1.8) --(9.5,-1);
\fi
\end{scope}


\begin{scope}[yshift=-16cm]
\node[NodeL] (leftset) at (-6,0){};
\node[anchor=south] at (leftset.90) {\LARGE  $L =\zo^n$};
\node[anchor=east] at (leftset.180) {\LARGE  {\bf $G_C$}\quad\quad};
\node[NodeLL] (leftset) at (-6,-1){};
\node[anchor=south] at (leftset.90) {\LARGE  $B_3$};
\node[] (leftnode) at (-6,-1){};
\node[anchor=south] at (leftnode.90) {\LARGE $\xc$};

\begin{scope}[xshift=9.5cm]
  
\node[NodeRR](bigR) at (0,0){};
\node[anchor=south] at (bigR.90) {\LARGE  $R =\zo^{\nc+\gamma(n/\delta)}$};

 \end{scope}

\draw[](-7.15,-2.5)--(-4.75,-2.5);
\draw[dashed](-6.9,-2.5)--(-6.9,-3.2);
\draw[dashed](-6.6,-2.5)--(-6.6,-3.6);
\draw[dashed](-6.3,-2.5)--(-6.3,-3.9);
\draw[dashed](-6.0,-2.5)--(-6.0,-4.2);
\draw[dashed](-5.7,-2.5)--(-5.7,-3.9);
\draw[dashed](-5.4,-2.5)--(-5.4,-3.6);
\draw[dashed](-5.1,-2.5)--(-5.1,-3.2);


\node[] at (-6,-1) (x){};
\node[NodeR] at (9.5,1.3) (z1){};
\node[NodeR] at (9.5,0.8) (z2){};
\node[NodeR] at (9.5,0.3) (z3){};
\node[anchor=west] at (z3) {\LARGE \,\,$\pc$};
\node[NodeRshared] at (9.5,-0.2) (z4){};
\node[] at (-6,-1.8) (dummy){};
\path[] (x) edge [bend left=10] (z1);
\path[] (x) edge [bend left=5] (z2);
\path[] (x) edge [bend left=0] (z3);
\path[] (x) edge [bend right= 5] (z4);
\path[tline] (dummy) edge [bend right = 15] (z4);
\if01
\draw[] (-6,-1) -- (9.5,1);
\draw [] (-6,-1) -- (9.5,0.5);
\draw [] (-6,-1) -- (9.5,0);
\draw [] (-6,-1) -- (9.5,-0.5);
\draw [] (-6,-1) -- (9.5,-1);
\draw[tline] (-6,-1.8) --(9.5,-1);
\fi
\end{scope}

\end{tikzpicture}
\caption{The encoding/decoding procedures. The senders use graphs $\ga, \gb,\gc$ with the rich owner property, and then encode $\xa, \xb, \xc$ by random neighbors $\pa,\pb,\pc$. The receiver uses $B_1, B_2, B_3$ in the small regime in the respective graphs, for which $\xa, \xb, \xc$ are rich owners, and reconstructs $\xa,\xb,\xc$ as the unique neighbors of $\pa, \pb, \pc$ in $B_1,B_2, B_3$.}
\label{f:encode}

\end{figure}

We need to show that if $\na,\nb,\nc$ satisfy the inequalities~(\ref{e:constraint}), then it is possible to reconstruct $(\xa,\xb,\xc)$ from $(\pa,\pb,\pc)$ with high probability (over the random choice of 
 $(\pa,\pb,\pc)$). The general idea is to identify computable enumerable subsets $B_1, B_2, B_3$ of left nodes in the three graphs, which are in the ``small regime,"  and which contain respectively $\xa, \xb, \xc$ as rich owners. Then $\pa$ has $\xa$ as its single neighbor in $B_1$, and therefore $\xa$ can be obtained from $\pa$ by enumerating the elements of $B_1$ till we find one that has $\pa$ as a neighbor  ($\xb, \xc$ are obtained similarly).

We shall assume first that the decoding procedure $D$ knows the $7$-tuple $(C(x_V) \mid V \subseteq \{A,B,C\}, V \not= \emptyset)$, \ie, the complexity profile of $(\xa,\xb, \xc)$.

The proof has an inductive character, so let us begin by analyzing the case when there is a single sender, then when there are two senders, and finally when there are three senders.
\medskip

\textbf{1 Sender.}  We show how to reconstruct $\xa$ from $\pa$, assuming $C(\xa) \leq \na$. Let
\[
B_1 = \{x \in \zo^n \mid C(x) \leq C(\xa)\}.
\]
Since the size of $B_1$ is bounded by $2^{C(\xa)+1} \leq 2^{\na+1}$, it follows that $B_1$ is in the small regime in $\ga$. The number of poor owners with respect to $B_1$ in $\ga$ is at most $(1/n^2) \cdot 2^{C(\xa)+1} \approx  2^{C(\xa) - 2\log n}$, and it can be shown that any poor owner can be described by $C(\xa) - \Omega(\log n)$ bits (essentially by its rank in some fixed standard ordering of the set of poor owners). Therefore the complexity of a poor owner is strictly less than $C(\xa)$ and thus $\xa$ is a rich owner with respect to $B_1$, as needed to enable its reconstruction from $\pa$.
\medskip

\textbf{2 Senders.}  We show how to reconstruct $\xa, \xb$ from $\pa, \pb$, assuming $C(\xa \mid \xb) \leq \na, C(\xb \mid \xa) \leq \nb, C(\xa, \xb) \leq \na + \nb$.

If $\na \geq C(\xa)$, then $\xa$ can be reconstructed from $\pa$ as in the \emph{1 Sender} case. Next, since $\nb \geq C(\xb \mid \xa)$, $\xb$ can be reconstructed from $\xa$ and $\pb$, similar to the \emph{1 Sender} case.

So let us assume that $C(\xa) > \na$.  Let
\[
B_2 = \{x \in \zo^n \mid C(x \mid \pa) \leq C(\xb \mid \pa)\}.
\]
We  show below that (1) $B_2$ is in the small regime in $\gb$, and (2) that it can be effectively enumerated.  Since $\xb \in B_2$, and since it is a rich owner with respect to $B_2$ (by a similar argument with the one used for $\xa$ in the \emph{1 Sender} case), this implies that $\xb$ can be reconstructed from  $\pa, \pb$, and next $\xa$ can be reconstructed from $\xb$ and $\pa$,  as in the \emph{1 Sender} case (because $\na \geq C(\xa \mid \xb)$). 

It remains to prove the assertions (1) and (2) claimed above. We establish the following fact.
\begin{fact} 
\label{f:f1}
\begin{enumerate}
\item[(a)] $C(\pa) \eqp \na$,
\item[(b)]  $C(\pa, \xb) \eqp C(\xa, \xb)$.
\item[(c)] $C(\xb \mid \pa) \eqp C(\xa, \xb) - \na$.
\item[(d)] $C(\xb \mid  \pa) \leqp \nb$.
\end{enumerate}
\end{fact}
\begin{proof}
(a) By the same argument used above, $\xa$ is still a rich owner with respect to $B_1$, but $B_1$ is now in the large regime. This implies that with probability $1- (1/n^2)$, $\pa$ has, for some constant $c$, $2^{C(\xa )- \na + c \log n}$ neighbors in $B_1$, one of them being $\xa$. The string $\xa$ can be constructed from $\pa$ and its rank among $\pa$'s  neighbors in $B_1$. This implies $C(\xa) \leqp  C(\pa) + (C(\xa) - \na)$, and thus, $C(\pa) \geqp \na$. Since $| \pa | \leqp \na$,  it follows that $C(\pa) \leqp \na$, and therefore, $C(\pa) \eqp \na$. 
\smallskip

(b)  The ``$\leqp$" inequality holds because $\pa$ can be obtained from $\xa$ and $O(\log n)$ bits which describe the edge $(\xa,\pa)$ in $\ga$. For the ``$\geqp$" inequality, let
\[
B_1' = \{x \in \zo^n \mid C(x\mid \xb) \leq C(\xa \mid \xb)\}.
\]
$B_1'$ is in the small regime in $\ga$ (because $|B_1'| \leq 2^{C(\xa\mid \xb)+1} \leq 2^{\na+1}$), and $\xa$ is a rich owner with respect to $B_1'$ (because, as we have argued above,  poor owners, being few,  have complexity conditioned by $\xb$, less than $C(\xa  \mid \xb)$). So, $\xa$ can be constructed from $\xb$ (which is needed for the enumeration of $B_1'$)  and $\pa$, and therefore, $C(\xa,\xb) \leq C(\pa, \xb)$.
\smallskip

(c)  $C(\xb \mid \pa) \eqp C(\pa, \xb) - C(\pa) \eqp C(\xa,\xb) - \na$, by using the chain rule, and (a) and (b).
\smallskip

(d) $C(\xb \mid \pa) \eqp C(\xa,\xb)- \na \leqp (\na + \nb) - \na = \nb$, by using (c) and the hypothesis.
\end{proof}
Now the assertions (1) and (2) follow, because by  Fact~\ref{f:f1}, (d), $B_2$ is in the small regime, and  by  Fact~\ref{f:f1} (c), the decoding procedure can enumerate $B_1$ since it knows $C(\xa,\xb)$ and $\na$.

Finally, we move to the case of three senders.
\smallskip

\textbf{3 Senders.} This is the case stated in Theorem~\ref{t:kolmslepwolf}. We show how to reconstruct $\xa,\xb,\xc$ from $\pa,\pb,\pc$,  if $\na,\nb, \nc$ satisfy the inequalities~(\ref{e:constraint}).  We can actually assume that $C(\xa) > \na, C(\xb) > \nb, C(\xc)> \nc$, because otherwise, if for example $C(\xa) \leq \na$, then $\xa$ can be reconstructed from $\pa$ as in the \emph{1 Sender} case, and we have reduced to the case of two senders. As in the \emph{2 Senders} case, it can be shown that $C(\pa) \eqp \na,  C(\pb) \eqp \nb, C(\pc) \eqp \nc$.

There are two cases to analyze.
\smallskip

\textbf{Case 1.} $C(\xb \mid \pa) \leq \nb$ or $C(\xc \mid \pa) \leq \nc$.   Suppose the first relation holds. Then $\xb$ can be reconstructed from $\pa, \pb$, by taking the small regime set for $\gb$,
\[
B_2' = \{x \in \zo^n \mid C(x \mid \pa) \leq \nb\},
\]
for which $\xb$ is a rich owner, and,  therefore, with high probability owns $\pb$. In this way, we reduce to the \emph{2 Senders} case.
\smallskip

\textbf{Case 2.} $C(\xb \mid \pa) >  \nb$ and $C(\xc \mid \pa) >  \nc$.  We show the following fact.
\begin{fact}
\label{f:f2}
\begin{enumerate}
\item[(a)]  $C(\xb \mid \xc, \pa) \leqp \nb$ and $C(\xc \mid \xb, \pa) \leqp \nc$ ,
\item[(b)]  $C(\xb,\xc \mid \pa) \leqp \nb + \nc$.
\end{enumerate}
\end{fact}
\begin{proof}
(a) First note that $C(\xc,\pa) \eqp C(\pa) + C(\xc \mid \pa) \geqp \na + \nc$.  Then
\[
\begin{array}{ll}
C(\xb \mid \xc, \pa) &\eqp C(\xb,\xc,\pa)- C(\xc, \pa) \\
&\leqp C(\xa,\xb,\xc) - C(\xc,\pa) \\
&\leqp (\na+\nb+\nc) - (\na+\nc) = \nb.
\end{array}
\]
The other relation is shown in the obvious similar way.
\smallskip

(b)
\[
\begin{array}{ll}
C(\xb, \xc \mid \pa) &\eqp C(\xb,\xc,\pa)- C(\pa) \\
&\leqp C(\xa,\xb,\xc) - C(\pa) \\
&\leqp (\na+\nb+\nc) - \na =  \nb + \nc.
\end{array}
\]
\end{proof}
Fact~\ref{f:f2} shows that,  given $\pa$, the complexity profile of $\xb, \xc$ satisfies the requirements for the \emph{2 Senders} case, and therefore these two strings can be reconstructed from $\pa, \pb, \pc$. Next, since $\na \geq C(\xa \mid \xb, \xc)$, $\xa$ can be reconstructed from $\pa, \xb, \xc$, as in the \emph{1 Sender} case. 

Thus,  both in  Case 1 (which actually consists of two subcases Case 1.1 and Case 1.2)  and in Case 2, $\xa,\xb, \xc$ can be reconstructed. There is still a problem: the decoding procedure needs to know which of the cases actually holds true. In the reference~\cite{zim:c:kolmslepianwolf}, it is shown how to determine which case holds true, but here  we present a solution that avoids this.

The decoding procedure launches parallel subroutines according to all three possible cases. The subroutine that works in the scenario that is correct  produces $(\xa, \xb, \xc)$. The subroutines that work with incorrect scenarios may produce other strings, or may not even halt.  How can this troublesome situation be solved? Answer: by hashing. The senders, in addition to $\pa,\pb,\pc$, use a hash function $h$, and send $h(\xa), h(\xb), h(\xc)$. The decoding procedure, whenever one of the parallel subroutines outputs a $3$-tuple, checks if the hash values of the tuple match $(h(\xa), h(\xb), h(\xc))$, and stops and prints that output the first time there is a match. In this way, the decoding procedure will produce with high probability $(\xa, \xb, \xc)$. 

\textbf{Hashing.} For completeness, we present one way of doing the hashing. By the Chinese Remainder Theorem, if $u_1$ and $u_2$ are $n$-bit numbers (in binary notation), then $u_1 \bmod p = u_2 \bmod p$ for at most $n$ prime numbers $p$. Suppose there are $s$ numbers $u_1, \ldots, u_s$, having length $n$ in binary notation and we want to distinguish the hash value of $u_1$ from the hash values of $u_2, \ldots, u_s$ with probability $1-\epsilon$.  Let $t= (1/\epsilon) sn$ and consider the first $t$ prime numbers $p_1, \ldots, p_t$. Pick $i$ randomly in $\{1, \ldots, t\}$ and define $h(u) = (p_i, u \bmod p_i)$. This \emph{isolates} $u_1$ from $u_2, \ldots, u_s$, in the sense that, with probability $1-\epsilon$, $h(u_1)$ is different from any of $h(u_2), \ldots, h(u_s)$. Note that the length of $h(u)$ is $O(\log n + \log s + \log(1/\epsilon))$.  In our application above, $s=3$, corresponding to the three parallel subroutines, and $\epsilon$ can be taken to be $1/n^2$, and thus the overhead introduced by hashing is only $O(\log n)$ bits. 

\textbf{Removing the assumption that the decoding procedure knows the inputs' complexity profile.} So far, we have assumed that the decoding procedure knows the complexity profile of $\xa, \xb, \xc$. This assumption is lifted using a hash function $h$, akin to what we did above to handle the various cases for \emph{3 Senders}. The complexity profile $(C(x_V) \mid V \subseteq \{A,B,C\}, V \not= \emptyset\}$ is a $7$-tuple, with all components bounded by $O(n)$. The decoding procedures launches $O(n^7)$ subroutines performing the decoding operation with known complexity profile, one for each possible value of the complexity profile. The subroutine using the correct value of the complexity profile will output $(\xa,\xb,\xc)$ with high probability, while the other ones may produce different $3$-tuples, or may not even halt. Using the hash values $h(\xa), h(\xb), h(\xc)$ (transmitted by senders together with $\pa, \pb, \pc$), the decoder can identify the correct subroutine in the same way as presented above. The overhead introduced by hashing is $O(\log n)$.

\section{Constructing graphs with the rich owner property}
\label{s:graph}

We sketch the construction needed for the \emph{computable} graph in Theorem~\ref{t:richownergraph}.  Recall that we use bipartite graphs of the form $G = (L=\zo^n, R=\zo^m, E\subseteq L \times R)$, in which every left node has degree $D=2^d$, and for every $x \in L$, the edges outgoing from $x$ are labeled with strings from $\zo^d$.  The construction relies on randomness extractors, which have been studied extensively in computational complexity and the theory of pseudorandom objects. A graph of the above type is said to be a $(k, \epsilon)$ extractor if for every $B \subseteq L$ of size $|B| \geq 2^k$ and for every $A \subseteq R$,
 \begin{equation}
 \label{e:eq1}
 \bigg | \frac{|E(B,A)|}{|B| \cdot D} - \frac{|A|}{|R|} \bigg | \leq \epsilon,
 \end{equation}
 where $|E(B,A)|$ is the number of edges between vertices in $B$ and vertices in $A$. In a $(k,\epsilon)$ extractor, any subset $B$ of left nodes of size at least $2^k$, ``hits" any set $A$ of right nodes like a random function: The fraction of edges that leave $B$ and land in $A$ is close to the density of $A$ among the set of right nodes. It is not hard to show that this implies the rich owner property in the large regime. To handle the small regime, we need graphs that maintain the extractor property when we consider prefixes of right nodes. 
  Given a bipartite graph $G$ as above and $m' \leq m$, the $m'$-prefix graph $G'$ is obtained from $G$ by  merging right nodes that have the same prefix of length $m'$. More formally, $G' = (L =\zo^n, R' =\zo^{m'}, E' \subseteq L \times R')$ and $(x,z') \in E'$ if and only if  $(x,z) \in E$ for some extension $z$ of $z'$. Recall that we allow multiple edges between two nodes, and therefore the merging operation does not decrease the degree of left nodes. 
 \begin{lemma}
 \label{l:extrand}
  For every $k \leq n$, and every $\epsilon > 0$, there exists a constant $c$ and a computable graph $G = (L = \zo^n, R=\zo^{k}, E \subseteq L \times R)$ with left degree $D =  cn/\epsilon^2$  such that for every $k' \leq k$, the $k'$-prefix graph $G' = (L =\zo^n, R' = \zo^{k'}, E' \subseteq L \times R')$  is a $(k', \epsilon)$ extractor.
 \end{lemma}
 \begin{proof} We  show the existence of a graph with the claimed properties using the probabilistic method. Once we know that the graph exists, it can be constructed by exhaustive search.  For some constant $c$ that will be fixed later, we consider a random function
 $f : \zo^n \times \zo^d \mapping \zo^{k}$.  This defines the bipartite graph $G = (L =\zo^n, R= \zo^{k}, E \subseteq L \times R)$ in the following way: $(x,z)$ is an edge labeled by $y$ if $f(x,y)=z$.
For the analysis, let us fix $k' \in \{1,\ldots, k\}$ and let us consider the graph $G' = (L, R', E' \subseteq L \times R')$ that is the $k'$-prefix of $G$. Let $K'=2^{k'}$ and $N=2^n$. Let us consider $B \subseteq \zo^n$ of size $|B| \geq K'$, and $A \subseteq R'$.
For a fixed $x \in B$ and $y \in \zo^d$, the probability that the $y$-labeled edge outgoing from $x$  lands in $A$ is $|A|/|R'|$. By the Chernoff bounds,
\[
\prob \bigg [ \bigg | \frac{|E'(B, A)|}{|B| \cdot D} -\frac{|A|}{|R'|} \bigg | > \epsilon \bigg ] \leq 2^{-\Omega(K' \cdot D \cdot \epsilon^2)}.
\]
The probability that relation~(\ref{e:eq1}) fails for  some $B \subseteq \zo^{k'}$ of size $|B| \geq K'$  and some $A \subseteq R'$ is bounded by $2^{K'} \cdot {N \choose K'} \cdot 2^{-\Omega(K' \cdot D \cdot \epsilon^2)}$, because $A$ can be chosen in $2^{K'} $ ways, and we can consider that $B$ has size exactly $K'$ and there are ${N \choose K'}$ possible choices of such $B$'s. If $D = cn/\epsilon^2$ and $c$ is sufficiently large, the above probability is  less than $(1/4)2^{-k'}$.  Therefore the probability that relation (\ref{e:eq1})  fails for some $k'$, some $B$ and some $A$ is less than $1/4$. It follows, that there exists a graph that satisfies the hypothesis. 
\end{proof}

Let $G = (L,R, E \subseteq L \times R)$ be the $(k,\epsilon)$-extractor from Lemma~\ref{l:extrand}.  Let  $\delta = (2\epsilon)^{1/2}$.  As hinted in our discussion above, by manipulating relation~(\ref{e:eq1}), we can show that for every $B \subseteq L$ of size $|B| >  2^k$, $(1-\delta)$ fraction of nodes in $B$ are rich owners with respect to $B$. This proves the rich owner property for sets in the large regime. If $B$ is in the small regime, then $B$ has size $2^{k'}$ for some $k' < k$ (for simplicity, we assume that the size of $B$ is a power of two). Let us consider $G'$, the $k'$-prefix of $G$. As above, in $G'$, $(1-\delta)$ fraction of elements $x$ in $B$ are rich owners with respect to $B$.  Recall that this means that if $x$ is a rich owner then $(1-\delta)$ fraction of its neighbors have $B$-degree bounded by $s$, where $s=(2/\delta^2) |B|\cdot D/ |R'| = O(n/\epsilon^3)$. Using the same hashing technique, we can ``split" each edge into $\poly(n/\epsilon)$ new edges.  More precisely,  an edge $(x,z)$ in $G'$ is transformed into $\ell = (1/\delta)sn $ new edges,
$(p_1, x \bmod  p_1, z), \ldots,  (p_\ell, x \bmod  p_\ell, z)$, where, as above,  $p_i$ is the $i$-th prime number. 
If $x$ is a rich owner then $(1-2\delta)$ of its ``new" neighbors (obtained after splitting) have $B$-degree equal to one, as desired, because hashing isolates $x$ from the other neighbors of $z$.  The $B$-degree of these right nodes continues to be one also in $G$, because when merging nodes to obtain $G'$ from $G$, the right degrees can only increase. Note that the right nodes in $G$ have as labels the $k$-bit strings, and after splitting we need to add to the labels the hash values $(p_i, x \bmod p_i)$, which are of length $O(\log n/\epsilon)$, and this is the cause for the $O(\log n/\delta)$ overhead in Theorem~~\ref{t:richownergraph}.

The \emph{explicit} graph in Theorem~\ref{t:richownergraph} is obtained in the same way, except that instead of the ``prefix" extractor from Lemma~\ref{l:extrand} we use the Raz-Reingold-Vadhan extractor~\cite{rareva:c:extractor}.

\section{Note}
This paper is dedicated to the memory of Professor Solomon Marcus. In the 1978 freshman Real Analysis class at the University of Bucharest, he asked several questions (on functions having pathological properties regarding finite variation). This has been my first contact with him, and, not coincidentally, also the first time I became engaged in a type of activity that resembled mathematical research. Over the years, we had several discussions, on scientific but also on rather mundane issues, and  each time, without exception, I was stunned by his encyclopedic knowledge on  diverse topics, including of course various fields of mathematics, but also literature, social sciences, philosophy, and whatnot.

\begin{quote}
The utilitarian  function of mathematics is in most cases a consequence of its cognitive function, but the temporal distance between the cognitive moment and the utilitarian one is usually imprevisible.

Solomon Marcus
\end{quote}



\begin{thebibliography}{CRVW02}

\bibitem[BZ14]{bau-zim:c:linlist}
Bruno Bauwens and Marius Zimand.
\newblock Linear list-approximation for short programs (or the power of a few
  random bits).
\newblock In {\em {IEEE} 29th Conference on Computational Complexity, {CCC}
  2014, Vancouver, BC, Canada, June 11-13, 2014}, pages 241--247. {IEEE}, 2014.

\bibitem[Cha66]{cha:j:length-of-programs}
G.~Chaitin.
\newblock On the length of programs for computing finite binary sequences.
\newblock {\em Journal of the ACM}, 13:547--569, 1966.

\bibitem[Cov75]{cov:j:slepwolfergodic}
Thomas~M. Cover.
\newblock A proof of the data compression theorem of {S}lepian and {W}olf for
  ergodic sources (corresp.).
\newblock {\em {IEEE} Transactions on Information Theory}, 21(2):226--228,
  1975.

\bibitem[CRVW02]{cap-rei-vad-wig:c:conductors}
M.~R. Capalbo, O.~Reingold, S.~P. Vadhan, and A.~Wigderson.
\newblock Randomness conductors and constant-degree lossless expanders.
\newblock In John~H. Reif, editor, {\em STOC}, pages 659--668. ACM, 2002.

\bibitem[DW85]{due-wol:j:slep-wolf-individ}
G.~Dueck and L.~Wolters.
\newblock The {S}lepian-{W}olf theorem for individual sequences.
\newblock {\em Problems of Control and Information Theory}, 14:437--450, 1985.

\bibitem[Kol65]{kol:j:kolmcomplexity}
A.N. Kolmogorov.
\newblock Three approaches to the quantitative definition of information.
\newblock {\em Problems Inform. Transmission}, 1(1):1--7, 1965.

\bibitem[Kuz09]{kas:j:slep-wolf-individ}
S.~Kuzuoka.
\newblock {S}lepian-{W}olf coding of individual sequences based on ensembles of
  linear functions.
\newblock {\em IEICE Trans. Fundamentals}, E92-A(10):2393--2401, 2009.

\bibitem[LZ76]{lem-ziv:j:compress}
A.~Lempel and J.~Ziv.
\newblock On the complexity of finite sequences.
\newblock {\em IEEE Trans. Inf. Theory}, IT-22:75--81, 1976.

\bibitem[Muc02]{muc:j:condcomp}
Andrei~A. Muchnik.
\newblock Conditional complexity and codes.
\newblock {\em Theor. Comput. Sci.}, 271(1-2):97--109, 2002.

\bibitem[Rom05]{rom:j:slepwolf}
A.~Romashchenko.
\newblock Complexity interpretation for the fork network coding.
\newblock {\em Information Processes}, 5(1):20--28, 2005.
\newblock In Russian. Available in English as~\cite{rom:t:slepwolf}.

\bibitem[Rom16]{rom:t:slepwolf}
Andrei Romashchenko.
\newblock Coding in the fork network in the framework of {K}olmogorov
  complexity.
\newblock {\em CoRR}, abs/1602.02648, 2016.

\bibitem[RR99]{raz-rei:c:extcon}
Ran Raz and Omer Reingold.
\newblock On recycling the randomness of states in space bounded computation.
\newblock In Jeffrey~Scott Vitter, Lawrence~L. Larmore, and Frank~Thomson
  Leighton, editors, {\em STOC}, pages 159--168. ACM, 1999.

\bibitem[RRV99]{rareva:c:extractor}
R.~Raz, O.~Reingold, and S.~Vadhan.
\newblock Extracting all the randomness and reducing the error in {T}revisan's
  extractor.
\newblock In {\em Proceedings of the 30th ACM Symposium on Theory of
  Computing}, pages 149--158. ACM Press, May 1999.

\bibitem[Sol64]{sol:j:inductive}
R.~Solomonoff.
\newblock A formal theory of inductive inference.
\newblock {\em Information and Control}, 7:224--254, 1964.

\bibitem[SW73]{sle-wol:j:distribcompression}
D.~Slepian and J.K. Wolf.
\newblock Noiseless coding of correlated information sources.
\newblock {\em {IEEE} Transactions on Information Theory}, 19(4):471--480,
  1973.

\bibitem[WZ94]{wyn-ziv:j:compress}
A.D. Wyner and J.~Ziv.
\newblock The sliding window {L}empel-{Z}iv is asymptotically optimal.
\newblock {\em Proc. IEEE}, 2(6):872--877, 1994.

\bibitem[Zim17]{zim:c:kolmslepianwolf}
M.~Zimand.
\newblock {K}olmogorov complexity version of {S}lepian-{W}olf coding.
\newblock In {\em STOC 2017}, pages 22--32. ACM, June 2017.

\bibitem[Ziv78]{ziv:j:compressindivid}
J.~Ziv.
\newblock Coding theorems for individual sequences.
\newblock {\em IEEE Trans. Inform. Theory}, IT-24:405--412, 1978.

\bibitem[Ziv84]{ziv:j:compresshelper}
J.~Ziv.
\newblock Fixed-rate encoding of individual sequences with side information.
\newblock {\em IEEE Trans. Inform. Theory}, IT-30:348--352--412, 1984.

\end{thebibliography}


\end{document}